\documentclass{article}

\usepackage{graphicx}
\usepackage{amssymb}

\usepackage{amsmath}
\usepackage{cite}
\usepackage{enumerate}

\newtheorem{thm}{Theorem}[section]
 \newtheorem{cor}{Corollary}[section]
 \newtheorem{lem}{Lemma}[section]
 \newtheorem{prop}{Proposition}[section]
 \newtheorem{conj}{Conjecture}[section]

\newtheorem{rem}{\bf{Remark}}[section]
 \newtheorem{ex}{\bf{Example}}[section]

\newcommand{\eh}{\hfill}\newlength{\sperr}

\newenvironment{proof}{{\settowidth{\sperr}{\bf\rm
Proof}%
\par\addvspace{0.3cm}\noindent\parbox[t]{1.3\sperr}
{\textit{ P\eh r\eh o\eh o\eh f\eh .}}%
}}{\nopagebreak\mbox{}
$\Box$\par\addvspace{0.3cm}}

\def\wt{\widetilde}

\title{Comparison of Thermodynamic Characteristics in 
Ordinary Quantum and Classical Approaches  and Game Theory}

\author{Lev Sakhnovich}

\date{}
\begin{document}
\maketitle
\begin{abstract}
 We fix the temperature $T$ and consider  mean energy
 and Boltzmann-Gibbs-Shannon entropy as two players of a game. 
As a result, basic formulas for the ordinary quantum mean energy and the Boltzmann-Gibbs-Shannon entropy are derived. 
We compare also the quantum and classical approaches  without a demand for Plank's constant being small.
 Important inequalities for statistical sum, quantum energy,
 quantum entropy, and their classical analogs follow.
\end{abstract}

Address: 99 Cove  ave., Milford, CT, 06461, USA \\
lsakhnovich@gmail.com




 \section{Introduction}\label{sec1}
 
 In the theory of ordinary quantum equilibrium systems \cite{2,4} the statistical sum
 \begin{equation}\label{1.1}
 Z_{q}(\beta, h)=\sum_{n=1}^{\infty}e^{-{\beta}E_{n}(h)},\quad \beta=1/kT
 \end{equation}
 plays the main role. In formula (\ref{1.1})  $k$ is the Bolzmann constant, $h$ is the Plank
constant, and $E_{n}(h)$ are eigenvalues of the energy operator $L$ (see
 (\ref{2.6})) of the considered system. In the classical physics the integral
 \begin{equation} \label{1.2}
 Z_{c}(\beta)=\int\int{e^{-{\beta}H(p,q)}\, dp, dq} \end{equation}
 is the analog of the sum in  (\ref{1.1}). In formula (\ref{1.2}) the function
 \begin{equation} \label{1.3}
 H(p,q)=\frac{1}{2m}\sum_{j=1}^{N}p_{j}^{2}+V(q)
 \end{equation}
 is the classical Hamiltonian, $p$ are corresponding generalized momenta, $q$ are
generalized coordinates, and $m$ is the mass of a particle.
  E. Wigner and J.G. Kirkwood (see \cite[Ch. 4]{3}) showed that the quantum statistical sum $Z_{q}(\beta,h)$ and
the classical statistical sum $Z_{c}(\beta)$ are connected by the relation
 \begin{equation} \label{1.4}
 \lim_{h{\to}0}(2{\pi}h)^{N}Z_{q}(\beta,h)=Z_{c}(\beta),
 \end{equation}
 where $N$ is the dimension of the corresponding coordinate space.
Schr\"odinger equation is of great current interest
 (see, for instance, recent  works 
\cite{0, d-1, d1, d2, d3, d4} and references therein).
In particular, the comparison of the
quantum and classical approaches without the demand for $h$ being small
 is of important scientific and methodological interest. For that, we consider the quantum mean energy
 \begin{equation} \label{1.5}
 E_{q}(\beta,h)=\sum_{n=1}^{\infty}E_{n}(h)e^{-{\beta}E_{n}(h)}/Z_{q}(\beta,h)
\end{equation}
and the classical mean energy
\begin{equation} \label{1.6}
 E_{c}(\beta)=\int{\int}H(p,q)e^{-{\beta}H(p,q)}dpdq/Z_{c}(\beta)
\end{equation}
of the same system. In the present paper we shall discuss the following conjectures.

\begin{conj}\label{Conjecture 1.1.}
The inequality
\begin{equation} \label{1.7}
 (2{\pi}h)^{N}Z_{q}(\beta,h){\leq}Z_{c}(\beta)
 \end{equation}
is true for all $h>0$ and $\beta>0$.
\end{conj}

\begin{conj}\label{Conjecture 1.2.}
The inequality
\begin{equation} \label{1.8}
E_{q}(\beta,h){\geq}E_{c}(\beta)
\end{equation}
is true for all $h>0$ and $\beta>0$.
\end{conj}

\begin{conj}\label{Conjecture 1.3.}
The following asymptotic equalities
\begin{equation} \label{1.9}
(2{\pi}h)^{N}Z_{q}(\beta,h){\thicksim}Z_{c}(\beta),\quad \beta{\to}+0,
\end{equation}
\begin{equation} \label{1.10}
E_{q}(\beta,h){\thicksim}E_{c}(\beta),\quad \beta{\to}+0
\end{equation}
are valid.
\end{conj}

It is essential that  the limits, which are considered in (\ref{1.4}) and (\ref{1.9}),
are different. Note also that $\beta=k/T$, which means that the relation
$\beta{\to}+0$ is equivalent to $T{\to}+\infty$.

We proved (see \cite{7, 8}) that in   important special cases (one-dimensional potential well and harmonic oscillator) relations(\ref{1.7})-(\ref{1.10}) are true. In the present paper we consider a much more general case.

In Section 5 we consider the quantum energy  E and the entropy S together. 
One could repeat a statement from  \cite[p. vii]{Ts}:
"... virtually nothing more basically than energy and entropy deserves the
qualification of pillars of modern physics".
The connection between E and S we interpret in terms of game theory.
The necessity of the game theory approach can be explained in the following
way. According to the second law of thermodynamics,
physical system in equilibrium has the maximal entropy
among all the states with the same energy.
So, a problem of the conditional extremum appears,
but the corresponding equation for the Lagrange multiplier is transcendental   and very complicated.
Therefore,  another argumentation is needed to find the
basic Gibbs formulas (see \cite{2,4}). 

In our approach we consider another extremal problem.
For that purpose we fix the Lagrange multiplier $\beta=k/T$
(not energy), that is, we fix the temperature and introduce the compromise function
$F={-\beta}E_{q}(\beta,h)+S_{q}(\beta,h)$. 
Then the mean energy $E_{q}(\beta,h)$ and the entropy $S_{q}(\beta,h)$ are two players of a game and 
the compromise result
is the extremum  point of $F$.
The obtained results have an interesting intersection with the results from game theory, namely, the transition from 
the classical(determined strategy) to quantum (probabilistic strategy) mechanics leads to a gain for both players.
Basic formulas for the quantum energy and for the entropy follow from this result.

\section{General case, statistical sum}\label{sec2}

 We use measure and integration connected with Wiener processes (see \cite{3, 5, 10}). With the help of these notions we formulate the important D. Ray's results \cite{6}.

 \begin{thm}\label{Theorem 2.1.}
  \textrm{(D. Ray \cite{6})}
  Let $\Omega$ be an open set in $R^{N}$ such that at each
 boundary  point $x$ of $\Omega$ there is a sphere with the center
 at $x$, some open sector
 of which is entirely outside the closure  $\overline{\Omega}$ of  $\Omega$. \\
 Let $\wt V(x)$ be a non-negative  function, which is bounded on each bounded subset of
  $\overline{\Omega}$ and satisfies a Lipschitz condition
  \begin{equation} \label{2.1}
   |\wt V(x^{\prime})-\wt V(x)|{\leq}M(x)|x^{\prime}-x|^{\alpha},\quad 0<\alpha{\leq}1,\quad
 x^{\prime},x\,{\in}\, \overline{\Omega}.\end{equation}
 Suppose also that either $\Omega$ is bounded or that
 \begin{equation} \label{2.2}
 \lim V(x)=\infty,\quad |x|{\to}\infty,\quad x{\in}\Omega.
 \end{equation}
 Then the differential operator
 \begin{equation} \label{2.3}
   \wt Lu=-\frac{1}{2}{\Delta}u+\wt V(x)u
   \end{equation}
on  $L^{2}(\overline{\Omega})$,  where $u=0$ on the boundary,  has a discrete spectrum
 $\lambda_{n}>0$ and
 \begin{equation} \label{2.4}
 \sum_{n=1}^{\infty}e^{-t \lambda_{n}}{\leq}
 (2{\pi}t)^{-N/2}\int_{\overline{\Omega}}e^{-t\wt V(x)}dx,\quad t>0.
 \end{equation}
 \end{thm}

 D. Ray proved also the relation
 \begin{equation} \label{2.5}
 \sum_{n=1}^{\infty}e^{-t \lambda_{n}}{\thicksim}
 (2{\pi}t)^{-N/2}\int_{\overline{\Omega}}e^{-t\wt V(x)}dx,\quad t{\to}+0.
 \end{equation}

 \begin{rem}\label{Remark 2.1.} 
 For the case $\Omega=R^{N}$, the results similar to
 (\ref{2.4}) and (\ref{2.5})  were deduced in
 a number of papers (see the results and references in B. Simon's  book \cite{10}, Chapter 3).
  \end{rem}

 The relation (\ref{2.5}) can be interpreted as a weak form of M. Kac's principle of
 imperceptibility of the boundary \cite{3} in the case of equation (\ref{2.3}).
 We note that our paper \cite{9} is dedicated to a weak form of the Kac's principle of imperceptibility of the boundary  in the case of the stable processes. It is
 interesting that D. Ray's results can be interpreted in a new way.

 We shall show that inequalities (\ref{1.6}) (Conjecture \ref{Conjecture 1.1.} ) and (\ref{1.8})
 (Conjecture \ref{Conjecture 1.3.}) follow from (\ref{2.4}) and (\ref{2.5}), respectively. To do it let us consider the Schr\"odinger differential operator
 \begin{equation} \label{2.6}
  L{\Psi}=-\frac{h^{2}}{2m}\Delta{\Psi}+V(x)\Psi
 \end{equation}
 in an open subset $\Omega$ of the $N$-dimensional Euclidean space $R^{N}$.
 Assume that
\begin{equation} \label{2.7}
 \Psi|_{\Gamma}=0,
 \end{equation}
where $\Gamma$ is the boundary of $\Omega$. Taking into account the relation
$\lambda_{n}=\frac{m}{h^{2}}E_{n}$ we see that in the case (\ref{2.6}) inequality (\ref{2.4}) has the
form
\begin{equation} \label{2.8}
 \sum_{n}e^{-t\frac{m}{h^{2}}E_{n}}{\leq}
 (2{\pi}t)^{-N/2}\int_{\overline{\Omega}}e^{-t\frac{m}{h^{2}}V(x)}dx,\quad t>0.
 \end{equation}
 From (\ref{1.2}), where the integral in $p$ is taken over
 $R^N$, and from (\ref{1.3}), we obtain
 \begin{equation} \label{2.9}
   Z_{c}(\beta)=
 (2{\pi}m/{\beta})^{N/2}\int_{\overline{\Omega}}e^{-{\beta}V(x)}dx.
 \end{equation}
 Putting $t=\frac{h^{2}}{m}{\beta}$ and taking into account (\ref{1.1}), we write (\ref{2.8}) in the form
 \begin{equation} \label{2.10}
  (2{\pi}h)^{N}Z_{q}(\beta ,h){\leq}Z_{c}(\beta),\quad \beta=k/T.
  \end{equation}
 In the same way we deduce from  (\ref{2.5}) that
  \begin{equation} \label{2.11}
   (2{\pi}h)^{N}Z_{q}(\beta ,h){\thicksim }Z_{c}(\beta),\quad\beta=k/T{\to}0.
  \end{equation}
  So we proved the following assertion.

  \begin{thm}\label{Theorem 2.2.}
 Let the conditions of Theorem \ref{Theorem 2.1.} be fulfilled. Then relations (\ref{1.7}) and (\ref{1.9})  are true.
 \end{thm}

 \begin{cor}\label{Corollary 2.1.}
Let the conditions of Theorem \ref{Theorem 2.1.} be fulfilled. If
  \begin{equation} \label{2.12}
    \int_{\overline{\Omega}}e^{-{\beta}V(x)}dx<\infty,
    \end{equation}
  then
  \begin{equation} \label{2.13}
   \sum_{n=1}^{\infty}e^{-\beta E_{n}}<\infty.\end{equation}
 \end{cor}

 \begin{ex}\label{Example 2.1.}
  \textbf{(Potential well.)}  
   If $\Omega$ is bounded and $V(x)=0$, then according to (\ref{2.9}) we have
  \begin{equation} \label{2.14}
   Z_{c}(\beta)=(2{\pi}m/{\beta})^{N/2}\mathrm{vol}(\Omega) ,
  \end{equation}
  where "vol" means volume.
\end{ex}

  \section{General case, mean energy}\label{sec3}

  The following assertion confirms partially Conjecture \ref{Conjecture 1.2.}.

  \begin{thm}\label{Theorem 3.1.}
Let the conditions of Theorem \ref{Theorem 2.1.} be fulfilled and
  \begin{equation} \label{3.1}
   E_{q}(\beta ,h)<\infty,\qquad E_{c}(\beta )<\infty .
   \end{equation}
  Then the inequality
  \begin{equation} \label{3.2}
   \int_{+0}^{\beta}\big(E_{q}(\gamma,h)-E_{c}(\gamma)\big)d{\gamma}{\geq}0,\qquad
  \beta>0
  \end{equation}
  is true.
  \end{thm}

  \begin{proof}
 Using relations (\ref{1.1}), (\ref{1.5}) and (\ref{1.2}), (\ref{1.6}) we have
  \begin{equation} \label{3.3}
  E_{q}(\beta,h)=-\frac{{\partial Z_{q}(\beta,h)}}{{\partial \beta}} \big/ Z_{q}(\beta ,h),\quad
 E_{c}(\beta)=-\frac{{\partial Z_{c}(\beta)}}{{\partial \beta}} \big/ Z_{c}(\beta).
 \end{equation}
 Formulas (\ref{3.3}) imply that
 \begin{equation} \label{3.4}
  \int_{\tau}^{\beta}\big(E_{q}(\gamma,h)-E_{c}(\gamma)\big)d{\gamma}=
 \log{\big( Z_{c}(\gamma)/(2{\pi}h)^{N}Z_{q}(\gamma,h)\big) } \Big|_{\tau}^{\beta},
  \end{equation}
  where $0<\tau <\beta$. According to (\ref{2.10}), (\ref{2.11}), and (\ref{3.4}) we obtain
  \begin{equation} \label{3.5}
   \int_{+0}^{\beta}[E_{q}(\gamma,h)-E_{c}(\gamma)]d{\gamma}=
 \log{\Big( Z_{c}(\beta)/(2{\pi}h)^{N}Z_{q}(\beta ,h)\Big) }{\geq}0.
 \end{equation}
 The theorem is proved.
 \end{proof}

 \begin{ex}\label{Example 3.1.}
\textbf{(Potential well.)}
  If $\Omega$ is bounded and $V(x)=0$, then according to (\ref{2.14}) and (\ref{3.3}) we have
  \begin{equation} \label{3.6}
    E_{c}(\beta)=N/2\beta.
    \end{equation}
\end{ex}

  \section{General case, Boltzmann-Gibbs-Shannon entropy}\label{sec4}

  The entropy of a quantum system is defined by the relation (see \cite{2, 4}):
  \begin{equation} \label{4.1}
   S_{q}=-\sum_{n=1}^{\infty}P_{n}{\log{P_{n}}},
   \end{equation}
  where
  \begin{equation} \label{4.2}
    P_{n}(\beta ,h)=e^{-{\beta}E_{n}(h)}\big/ Z_{q}(\beta ,h).
  \end{equation}
 The entropy is one of the fundamental notions in quantum mechanics
 (see some recent results, discussions, and references in 
 \cite{BeS, e0, e3, Ts, Weh1, Weh2}).
  It follows from (\ref{1.1}), (\ref{1.5}) and  (\ref{4.1}), (\ref{4.2}) that
  \begin{equation} \label{4.3}
   S_{q}(\beta,h)={\beta}E_{q}(\beta,h)+{\log{Z_{q}(\beta,h)}}.
   \end{equation}
  The classical definition of the entropy has the form (see \cite[Ch. 1]{4}):
  \begin{equation} \label{4.4}
  S_{c}(\beta,h)={\beta}E_{c}(\beta)+{\log{Z_{c}(\beta)}}-N{\log{(2{\pi}h)}}.
  \end{equation}
  Note that  (\ref{4.4}) contains a regularizing term
  $ -N{\log{(2{\pi}h)}}$.

  \begin{ex}\label{Example 4.1.}   
Now, we consider the case
  \begin{equation} \label{4.5}
    E_{n}(h)=\phi(h)E_{n},\quad E_{n}>0,
  \end{equation}
  where $E_{n}$ does not depend on $h$, $\phi(h)$ is a monotonically increasing
  function and $\phi(+0)=0$. In this case formula (\ref{4.3}) takes the form
  \begin{equation} \label{4.6}
     S_{q}(\beta ,h)=-{\lambda}\Phi^{\prime}(\lambda)\big/ \Phi(\lambda)+{\log{\Phi(\lambda)}},
  \end{equation}
  where $\lambda=\phi(h)\beta$ and
  \begin{equation} \label{4.7}
    \Phi(\lambda)=\sum_{n=1}^{\infty}e^{-{\lambda}E_{n}}.
  \end{equation}
\end{ex}

  \begin{lem}\label{Lemma 4.1.}
The function
  \begin{equation} \label{4.8}
 \Psi(\lambda)=-{\lambda}\Phi^{\prime}(\lambda)/\Phi(\lambda) +
  {\log}\, \Phi(\lambda)
  \end{equation}
is a monotonically decreasing function.
\end{lem}

\begin{proof}
 According to (\ref{4.7}) and (\ref{4.8}) the inequality
  \begin{equation} \label{4.9}
 \Psi^{\prime}(\lambda)=
  -{\lambda}\sum_{n>m}(E_{n}-E_{m})^{2}e^{-{\lambda}(E_{n}+E_{m})}/\Phi^{2}(\lambda)<0
  \end{equation}
is true.  The lemma is proved.
\end{proof}
Using Lemma \ref{Lemma 4.1.} we obtain the following statement.

\begin{thm}\label{Theorem 4.1.}
Let conditions (\ref{4.5}) be fulfilled.
Then the entropy  $S_{q}(\beta,h)$ is a monotonically decreasing function with respect to $\beta$ and with respect to $h$.
\end{thm}

\begin{ex}\label{Example 4.2.}
\textbf{(Potential well.)}
Consider the spectral problem
\begin{equation} \label{4.17}
  -\frac{h^{2}}{2m}\Delta{u}=Eu,\quad u_{\Gamma}=0
  \end{equation}
in an $N$-dimensional domain $\Omega$ with the boundary $\Gamma$. It is easy to see
that condition (\ref{4.5}) is fulfilled for this problem and
$\phi(h)=h^{2}$. Hence, the assertions of Theorem \ref{Theorem 4.1.} are true in the case of the potential well.
\end{ex} 

\begin{ex}\label{Example 4.3.}
\textbf{(Homogeneous potential.)}
We suppose that the non-negative
potential $V(x)$ is a homogeneous function, that is,
\begin{equation} \label{4.18}
  V(hx)=h^{\nu}V(x),\quad h>0,\quad \nu>0,\quad x{\in}R^{N}.
\end{equation}
\end{ex} 

\begin{rem}\label{Remark 4.3.} 
The function
\begin{equation} \label{4.19}
   V(r)=r^{\nu},\quad r=\Big(\sum_{k=1}^{N}x_{k}^{2}\Big)^{1/2},\quad \nu>0 \end{equation}
satisfies condition (\ref{4.18}).
\end{rem}

Using substitution $x=h^{\alpha}y$ we rewrite the equation
\begin{equation} \label{4.20}
  -\frac{h^{2}}{2m}\Delta{u(x)}+V(x)u(x)=Eu(x) \end{equation}
in the form
\begin{equation} \label{4.21}
-\frac{h^{2-2\alpha}}{2m}\Delta{\wt u(y)}+h^{\alpha{\nu}}V(y)\wt u(y)=E\wt u(y).
\end{equation}
We note that $2-2\alpha=\alpha{\nu}$ if $\alpha=2/(2+\nu)$. In this case
(\ref{4.5}) holds and
\begin{equation} \label{4.22}
\phi(h)=h^{2\nu/(2+\nu)}.
\end{equation}
 Hence, the assertions of Theorem \ref{Theorem 4.1.} are true in the case of homogeneous potentials.
 
Relations  (\ref{1.1}) and (\ref{1.5}) imply the assertion.

\begin{prop}\label{propd1} 
Let a potential V(r) have the form (\ref{4.19}).
Then we have
\begin{equation}\frac{d}{dh}\big(h^{N}Z_{q}(\beta,h)\big)=h^{N-1}Z_{q}(\beta,h)
\big(N-{\alpha\beta}E_{q}(\beta,h)\big).
\end{equation}
\end{prop}
By some direct calculations we get another proposition.

\begin{prop} \label{propd2} 
Let a potential V(r) have the form (\ref{4.19}).
Then we have
\begin{equation}E_{c}(\beta)=\frac{N}{\alpha\beta}.\end{equation}
\end{prop}
From Propositions \ref{propd1} and \ref{propd2} we obtain the statement below.

\begin{prop} \label{propd3}
Let a potential V(r) have  the form (\ref{4.19}). Then the  relations
\begin{equation}E_{q}(\beta,h)>E_{c}(\beta)
\quad {\mathrm{and}} \quad \frac{d}{dh}\big(h^{N}Z_{q}(\beta,h)\big)<0
\end{equation}
are equivalent.
\end{prop}
In addition to Conjecture \ref{Conjecture 1.1.} we  formulate the following conjecture.

\begin{conj}
\label{Conjecture 4.1.}
The function $h^{N}Z_{q}(\beta,h)$ is monotonically decreasing
with respect to $h$. 
\end{conj}

If (\ref{4.19}) holds, Proposition \ref{propd3}  shows that  Conjectures 
 \ref{Conjecture 1.1.} and  \ref{Conjecture 4.1.} are equivalent.

\begin{rem} \label{Rem4.2.} In particular, our considerations could be used
to treat an old problem by A. Wehrl. In his paper \cite{Weh2} he wrote:
"It is usually claimed that in the limit $h{\to}0$ the quantum-mechanical expression
tends toward the classical one, however, a rigorous proof of this is nowhere to be found in the literature".
Assuming that
\begin{eqnarray} \label{W1} &&
E_{q}(\beta,h){\to}E_{c}(\beta), \quad h{\to}0, 
\\
\label{W2}&&
(2{\pi}h)^{N}Z_{q}(\beta,h){\to}Z_{c}(\beta), \quad h{\to}0,
\end{eqnarray}
and using (4.3) and (4.4) we get
\begin{equation}\label{W3}
S_{q}(\beta,h)=S_{c}(\beta,h)+o(1),\quad h{\to}0.
\end{equation}
The present paper contains conditions for relations (\ref{W1})
and (\ref{W2}) to hold.
\end{rem}

 \section {Connection between energy and entropy, game theory point of view}\label{sec5}
Let the eigenvalues $E_{n}$ of the energy operator $L$ be given. Consider the
mean energy
$E=\sum_{n}{E_{n}P_{n}}$ and  the entropy $S=-\sum_{n}P_{n}\log{P_{n}}$.
Here $P_{n}$ are the corresponding probabilities, that is,  $\sum_{n}P_{n}=1$.
Hence $P_{n}$ can be represented in the following form
$P_{n}=p_{n}/Z$, where $Z=\sum_{n}p_{n}$.
Our aim is to find the probabilities $P_{n}$.
For that purpose we consider the function
\begin{equation} \label{5.1}
  F=\lambda{E}+S,
\end{equation}
where $\lambda=-\beta=-1/kT$ (see (\ref{1.1})).\\

\noindent {\bf Fundamental Principle.} {\it The function $F$ defines the game between
the mean energy $E_q$ and the entropy $S_q$.}\\

To find the stationary point of $F$ we calculate
\begin{equation} \label{5.2}
  \frac{\partial{F}}{\partial{p_{k}}}=\lambda\Big(E_{k}/Z-\sum_{n=1}^{\infty}E_{n}p_{n}/Z^{2}\Big)-
(\log{p_{k})/Z}+\sum_{n=1}^{\infty}p_{n}\log{p_{n}}/Z^{2}.
\end{equation}
It follows from (\ref{5.2}) that the point
\begin{equation} \label{5.3}
  p_{n}=e^{{\lambda}E_{n}}, \qquad n=1,2,\ldots
\end{equation}
is a stationary point.  Moreover, the stationary point is unique up to a scalar
multiple. Without loss of generality this multiple can be fixed as in (\ref{5.3}).

\begin{cor}\label{Corollary 5.0.}
The basic formulas (\ref{1.5}), (\ref{4.1}), and (\ref{4.2}) are
immediate from (\ref{5.3}).
 \end{cor}

By direct calculation we get in the stationary point (\ref{5.3}) the equalities
\begin{equation} \label{5.4}
  \frac{\partial^{2}F}{\partial{p_{k}^{2}}}=-Z_k/(p_kZ^2)<0,
\quad Z_k:= \sum_{j\not= k}p_j;
\quad \frac{\partial^{2}F}{\partial{p_{k}}\partial{p_{j}}}=1/Z^{2}>0,\quad j{\ne}k .
\end{equation}
Relations (\ref{5.4}) imply the following assertion.

\begin{cor}\label{Corollary 5.1.}
 The stationary point  (\ref{5.3}) is a maximum of the function $F$.
 \end{cor}
\begin{proof} We use the following result (see \cite[Ch.7, Problem 7]{PS}):
\begin{equation}\label{5.5} \det\left[\begin{array}{ccccc}
                      r_{1} & a & a & ... & a \\
                      b & r_{2} & a & ... & a \\
                      b & b & r_{3} & ... & a \\
                      ... & ... & ... & ... & ... \\
                      b & b & b & ... & r_{k}
                    \end{array}\right]=\frac{af(b)-bf(a)}{a-b},\end{equation}
 where
 \begin{equation}\label{5.6} f(x)=( r_{1}-x)( r_{2}-x)...( r_{k}-x).\end{equation}
In the case that $a=b$, the equality below is easily derived from \eqref{5.5}:
\begin{equation}\label{5.7} \det\left[\begin{array}{ccccc}
                      r_{1} & a & a & ... & a \\
                      a & r_{2} & a & ... & a \\
                      a & a & r_{3} & ... & a \\
                      ... & ... & ... & ... & ... \\
                      a & b & a & ... & r_{k}
                    \end{array}\right]=-af^{\prime}(a)+f(a).\end{equation}
Using (\eqref{5.4} and \eqref{5.7} we calculate the Hessian $H_{k}(F)$ in the stationary point:
\begin{equation}\label{5.8}  H_{k}(F)=Z^{-2k}[-f^{\prime}(1)+f(1)],
\end{equation}
where $f$ is given by \eqref{5.6} and $r_{n}=-Z_{n}/p_{n}$. Rewrite \eqref{5.8} in the form
\[
H_{k}(F)=(-Z)^{-k}\Big(1-\big(\sum_{n=1}^k p_n\big)/Z\Big)/\prod_{n=1}^k p_n
\]
to see that the relation
${\mathrm{sgn}\,}\big(H_{k}(F)\big)=(-1)^{k}$
is true. Hence, the corollary is proved.
\end{proof}
Note that  the basic relations \eqref{5.3} 
are obtained by solving a new extremal problem. Namely, in the introduced function F
the parameter $\lambda$ is fixed  instead of the  energy $E$, which is usually fixed.

 \begin{rem}\label{Remark 5.1.}
In the game theory the transition from deterministic to probabilistic strategy leads to a gain for  players. 
The transition from classical to quantum mechanics leads to a gain for both players (energy and entropy) too (see (\ref{1.8}) and
Theorem \ref{Theorem 4.1.}).
\end{rem}

\section{Conclusion}
For small values of $h$ the relation between   quantum and
classical statistical sums was deduced by E. Wigner and J.G. Kirkwood.
However, the comparison
of the quantum and classical approaches for energy, statistical sum and entropy
without the demand of h being small is of essential scientific and methodological interest. In our paper we obtain some general results 
and discuss some conjectures connected with the formulated problem.

In particular, general and rigorous results on relations between ordinary quantum and classical statistical sums 
(see Theorem \ref{Theorem 2.2.}) could be derived
from an important work by D. Ray \cite{6}
on the spectra of Schr\"odinger operators. Furthermore, our approach allows 
to treat an old entropy problem by A. Wehrl (see Remark \ref{Rem4.2.}).
We introduce also the function, the extremum point of which gives the well-known Gibbs formulas.   
The results of the paper intersect with
some ideas of game theory
(see, e.g.,   \cite{i1}):   the transition from classical
 (determined strategy) to quantum mechanics (probabilistic strategy) leads to a gain for both players. 
However, we stress that the connection between energy and entropy is a new type of a game, where
the players do not have a freedom to choose their strategy.

Our note could be considered as an input into the important discussion
on  the deterministic and probabilistic aspects of quantum theory
(see \cite{AB, Groes, Holland1}, and references therein).
As the next step it would be fruitful to consider also the possibility to generalize
our conjectures and results for the case of  nonextensive statistical mechanics
\cite{Ts}.












\end{document}